\documentclass[conference]{IEEEtran}
\usepackage{graphicx,cite,epsfig,amssymb,amsmath}
\usepackage{pifont}
\usepackage{stmaryrd}
\usepackage{bbding}
\usepackage{subfigure}
\usepackage{algorithm}
\usepackage{algorithmic}
\usepackage{slashbox,multirow}
\usepackage{mathrsfs}
\usepackage{latexsym}
\usepackage{indentfirst}
\usepackage{fmtcount}
\usepackage{cite}

\newtheorem{theorem}{Theorem}

\newtheorem{lemma}{Lemma}

\IEEEoverridecommandlockouts

\begin{document}

\title{Rateless-Coding-Assisted Multi-Packet Spreading over Mobile Networks}

\author{\authorblockN{Huazi~Zhang$^{\dag\ddag}$, Zhaoyang~Zhang$^{\dag}$, Huaiyu~Dai$^{\ddag}$
\\$\dag$. Dept.of Information Science and Electronic Engineering, Zhejiang University, China.
\\$\ddag$. Department of Electrical and Computer Engineering, North Carolina State University, USA
\\Email: hzhang17@ncsu.edu, ning\_ming@zju.edu.cn, huaiyu\_dai@ncsu.edu}
}

%
\maketitle


\begin{abstract}
A novel Rateless-coding-assisted Multi-Packet Relaying (RMPR) protocol is proposed for large-size data spreading in mobile wireless networks. With this lightweight and robust protocol, the packet redundancy is reduced by a factor of $\sqrt n$, while the spreading time is reduced at least by a factor of $\ln (n)$. Closed-form bounds and explicit non-asymptotic results are presented, which are further validated through simulations. Besides, the packet duplication phenomenon in the network setting is analyzed for the first time.
\end{abstract}

\begin{keywords}
Rateless Coding, Mobile Networks, Multi-Packet Relaying, Information Spreading.
\end{keywords}

\section{Introduction}
\noindent From the dissemination of genetic information through replications of DNA, and the spread of rumors via Twitter, to the transfer of data packets among wireless devices over electromagnetic waves, the phenomenon of information spreading influences every aspect of our lives. In these scenarios, how fast information can be spread to the whole network is of particular interest.

Information spreading in static and connected networks has
been studied in literature \cite{Gossip}. Meanwhile, tremendous research
efforts \cite{Passarella,Eun,MEG,XiaolanZhang,GroeneveltThesis,Mobile-conductance,Mobile-conductance-TWC} (and the reference therein) have been made
in theoretically modeling both inter-contact time \cite{Passarella,Eun,MEG} and
message delay \cite{XiaolanZhang,GroeneveltThesis} in mobile networks, especially
the disconnected networks. Due to space limitation, more
complete descriptions of this topic can be found in the survey
\cite{Passarella} and in our technical report \cite{tech-report}.


Our focus is slightly different, i.e. practical solution for multiple packets broadcasting. Except for a unique \emph{source} with the entire information, each of the rest nodes plays the roles of both \emph{relay} and \emph{destination}. However, this setting incurs two fundamental issues. First, due to the diversity of relay paths, a particular packet may be unnecessarily received multiple times. Second, under the network randomness, it is difficult to guarantee the reception of certain packets without repeated requesting and acknowledging. Both issues, if unsolved, undermine the system efficiency greatly.

Rateless code \cite{ratelesscode} is a class of codes designed for highly lossy channels, e.g. deep space channel. The rateless encoder generates potentially an unlimited number of distinct packets, which prevents repeated packet receiving. Besides, the rateless decoder only requires adequate number of packets to be received, rather than acknowledging specific source packets. This packet-level acknowledge-free feature inspires us to develop a lightweight and robust protocol for spreading large-size data.

Our contributions are summarize as below.
\begin{enumerate}
\item We propose a simple and easy-to-implement \emph{Rateless-coding-assisted Multi-Packet Relaying (RMPR) protocol}, where the individual packets do not have to be acknowledged. Thus, the protocol efficiency is not compromised by the relatively long network delay.

\item The RMPR protocol enhances performance in terms of both message and time efficiency. The number of redundant packets received is reduced from $\sqrt n$ to 2. The source-to-destination delay and the source-to-network spreading time are also improved by at least a factor of $\ln (n)$.

\end{enumerate}

\section{Problem Formulation and System Model}\label{ModelSection}
\subsection{Homogeneous and Stationary Mobile Network (HSMNet)}
\noindent We study a mobile network $G\left( {V\left( t \right),E\left( t \right)} \right)$ which consists of $n$ nodes moving in a given area (e.g., a unit square), according to a certain mobility model. In this study we consider a general class of mobile networks, coined as the Homogeneous and Stationary Mobile
Network (HSMNet), which is characterized by the following three properties:

\begin{itemize}
\item It is assumed that the spatial distribution of each node has converged (after sufficient evolvement) to a stationary distribution, denoted by $\pi_i(x, y)$ for each node, where $(x, y)$ is the location in the area of interest.

\item $\pi_1(x, y)=...=\pi_n(x, y)\triangleq \pi(x, y)$, which means the network nodes are homogeneous.

\item $\pi(x, y)>0, \forall x,y$, which means every node can travel to any position on the given area given enough time. It is also assumed that all nodes move at a constant speed $v$; other than this, no more specification on the mobility pattern is needed.

\end{itemize}

 One crucial parameter for the mobile network model is the transmission range $r$ within which two nodes can exchange packets. Here the transmission is assumed to be instantaneous, and the range is assumed to be $r = \Theta \left( {\frac{1}{{\sqrt n }}} \right)$, which indicates a disconnected network. Otherwise, the spreading time would be always zero, making the problem trivial.




\subsection{Packet Transmission upon Meeting}
\noindent We adopt a continuous time system model \cite{GroeneveltThesis}, instead of the slotted model. The transmission of a packet is assumed to be instantaneous and error-free due to the small packet size, and only occur upon a \emph{meeting}, which is defined as the event that two nodes travel into each other's transmission range $r$ and exchange one packet.
Note that, the mobile nodes usually move very fast and the meeting duration is very short, such as in the Vehicular Ad-hoc Networks (VANET).

The information spreading is constituted of numerous message exchanges through the``meetings". The \emph{first-meeting time} is defined as the time interval between the an arbitrary chosen starting point and the first meeting. The \emph{inter-meeting time} is defined as the time interval between two consecutive meetings.

The meeting process between any two nodes with constant speed $v$ and transmission range $r \ll 1$ is shown to be a Poisson process \cite{GroeneveltThesis}. The first-meeting time and inter-meeting time between any two nodes in an HSMNet are \emph{exponentially distributed}, defined by the \emph{key parameter $\lambda$} given by \begin{align}\label{PoissonMeeting}
\lambda \approx \frac{8vr}{\pi}\int\limits_0^1 {\int\limits_0^1 {\pi ^2\left( {x,y} \right)dx} dy} = {c_\pi}vr,
\end{align}
where ${c_\pi} = \frac{8}{\pi }\int\limits_0^1 {\int\limits_0^1 {\pi ^2\left( {x,y} \right)dx} dy}$ is a constant determined solely by the stationary distribution of the HSMNet, which means $\lambda$ is only proportional to the node speed $v$ and transmission range $r$. In particular, in the well-known \emph{random direction mobility model} and \emph{random waypoint mobility model}, $\lambda$ equals $\frac{8vr}{\pi}$ and $\frac{8\omega vr}{\pi}$ respectively, where constant $\omega \approx 1.3683$ \cite{GroeneveltThesis}.

\section{Rateless-coding-assisted Multi-Packet Relaying Scheme}
\noindent The celebrated rateless coding \cite{ratelesscode} saves signaling as well as avoids packet duplication in point-to-point transmissions. In this work, we explore its application in the network setting. We will start with a naive packet spreading protocol, which leads to a duplication factor of $\sqrt{n}$ in packet relaying. We then present our RMPR protocol, and reveal through both theoretical analysis and simulation that the corresponding expected duplication factor is 2, independent of the network size and a dramatic increase in network efficiency. We further study the average delay for packet delivery to both an arbitrary destination node and the whole network, respectively.

In the context of rateless coding, any subset of coded packets of size $l' = l \times \left( 1 + \varepsilon \right)$ is sufficient to recover the original $l$ packets with high probability \cite{ratelesscode}, where $\varepsilon>0$ is a small constant. Therefore, we can \emph{count the number of distinct packets at the destination nodes, to determine whether the $l$ original source packets are recovered}.

\subsection{Spreading Protocol Description}
\subsubsection{A naive protocol}
we first discuss a simple rateless-coding-based protocol, and reveal the severity of the packet duplication problem in mobile relay networks.

\begin{itemize}
\item The \emph{source node} transmits a new packet upon every meeting with another node. Each \emph{relay node} simply retransmits \emph{any} packet received from the source node and other relay nodes. If a relay node has multiple packets, it \emph{randomly} picks one of them and transmits to another node upon each meeting.
\end{itemize}

The source node meets the rest nodes at a rate $\left(n-1\right)\lambda$, and so is the new packet growth rate for the network. We say a packet is non-redundant if it is never received before, and denote by $\kappa$ the probability for an arbitrary node to receive a non-redundant packet upon a meeting with an arbitrary relay node\footnote{All relay nodes and rateless packets are assumed homogeneous.}. Including the meetings with the source node, the overall non-redundant rate $\bar \kappa$ is given by
\begin{align*}
\bar \kappa = \frac {n-1} n \kappa + \frac 1 n.
\end{align*}

In the long run, $\bar \kappa$ converges to the ratio between the number of different packets and the number of total packet copies in the network$^1$. A copy is generated with probability one upon a meeting involving the source node, or with probability $\kappa$ upon a meeting between two relay/destination nodes. Thus, $\bar \kappa$ may also be given by
\begin{align*}
\bar \kappa = \frac{{\left( {n - 1} \right)\lambda t }}{ {{\kappa}\left( {n - 1} \right)\left( {n - 1} \right)\lambda } t + \left( {n - 1} \right)\lambda t}.
\end{align*}

By solving the above two equations, we get $\kappa  = \frac {\sqrt n -1} {n-1}$, which approximates $\frac{1}{{\sqrt n }}$ when $n$ goes large. This indicates that only one out of $\sqrt n$ received packets is a non-duplicate one.

\emph{Discussions:} the naive protocol allows multi-hop relaying, which inevitably introduces duplicated packets via multiple routing paths. The protocol essentially becomes inefficient as network size grows.

%
%
%
%

\subsubsection{The RMPR protocol}
we further propose a protocol that ensures a constant duplication rate.
\begin{itemize}
\item The \emph{source node} transmits a code packet each time it meets another node. Each packet it generates for transmission is unique and different from any packet that is already in the network.
\item A \emph{relay/destination node} can receive packets from both the source node and other relay nodes. Moreover, a relay node \emph{only} transmits the \emph{newest packet} that is \emph{directly} received from the source node.
\end{itemize}

\emph{Discussions:} by only transmitting the packets directly received from the source node, the multi-path issue is solved, and the duplicate packets only come through the single ``source-relay-destination" path. In addition, the number of packet copies in the network is non-decreasing over time. Therefore, the optimal choice to avoid duplication is that every relay node always picks the newest packet to retransmit.

\subsection{Exploring the Phenomenon of Duplicated Packet Reception}
\subsubsection{Theoretical Analysis}
\begin{figure}[h] \centering
\includegraphics[width=0.5\textwidth]{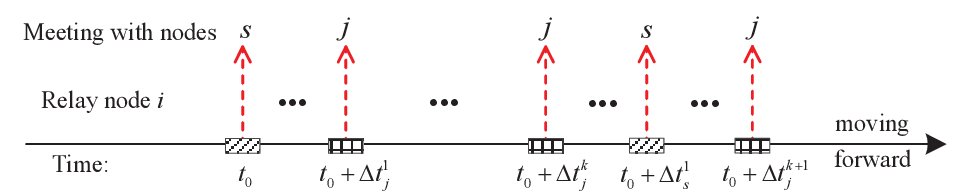}
\caption {Illustration of Received Packet Duplication}
\label{Duplication}
\end{figure}
Since a relay node only transmits the newest packet received from the source node, packet duplications only occur when some relay node, say $i$, meets the same destination node, say $j$, multiple times between two consecutive meetings with the source node.
Fig.~\ref{Duplication} illustrates the case that destination node $j$ receives the same packet $k$ times from node $i$, in which $k-1$ packets are duplications. In the figure, $t_0$ denotes one moment at which node $i$ meets the source node, while $t_0 + \Delta t_s^k$ and $t_0 + \Delta t_j^k$ denote the moments at which node $i$ meets the source node and node $j$ for the $k$th time, respectively.

According to the Poisson inter-meeting model, each of the intervals between two consecutive meetings of two nodes is i.i.d exponential. Thus,
\begin{align*}
\Delta t_s^1 \sim Exp\left( {\lambda } \right), \Delta t_j^k \sim Erlang\left( {k,\lambda } \right).
\end{align*}
Considering these properties, we can get the expected copies of duplicate packets as stated in the following Theorem:

\begin{theorem}[Duplication Analysis]\label{theorem-duplication}
Under the RMPR protocol,for any destination node, any packet is expected to be received two times on average.
\end{theorem}
\begin{proof}
According to Fig.~\ref{Duplication}, the probability that node $j$ receives the same packet $k$ times, given that the packet is received by node $j$ at least once, is
\begin{align}\label{k-probability}
&p\left( {\left. {\Delta t_j^k < \Delta t_s^1 < \Delta t_j^{k + 1}} \right|\Delta t_j^1 < \Delta t_s^1} \right)\nonumber\\
& = \frac{{p\left( {\Delta t_j^k < \Delta t_s^1 } \right) - p\left( {\Delta t_j^{k + 1} < \Delta t_s^1 } \right)}}{{p\left( {\Delta t_j^1 < \Delta t_s^1} \right)}}.
\end{align}

Denote by $d$ the average number of redundant copies of a certain packet, given by
\begin{align}\label{d_redundant}
d &= \sum\limits_{k = 2}^\infty  {p\left( {\left. {\Delta t_j^k < \Delta t_s^1  < \Delta t_j^{k + 1}} \right|\Delta t_j^1 < \Delta t_s^1 } \right)\left( {k - 1} \right)} \nonumber \\
& = \frac{{\sum\limits_{k = 2}^\infty  {p\left( {\Delta t_j^k < \Delta t_s^1 } \right)} }}{{p\left( {\Delta t_j^1 < \Delta t_s^1 } \right)}}.
\end{align}

Since $\Delta t_j^1$ and $\Delta t_s^1$ are i.i.d exponential with parameter $\lambda$,
\begin{align}\label{tj1}
p\left( {\Delta t_j^1 < \Delta t_s^1} \right) = \int\limits_0^\infty  {p\left( {\left. {\Delta t_j^1 < \tau _s^1} \right|\Delta t_s^1 = \tau _s^1} \right)d\tau _s^1} \nonumber \\
 = \int\limits_0^\infty  {\lambda {e^{ - \lambda \tau _s^1}}\left( {1 - {e^{ - \lambda \tau _s^1}}} \right)d\tau _s^1}  = \frac{1}{2}.
\end{align}

Similarly, since $\Delta t_j^k$ and $\Delta t_s^1$ are independent, \begin{align}\label{tjk}
p\left( {\Delta t_j^k < \Delta t_s^1} \right) &= \int\limits_0^\infty  {p\left( {\left. {\Delta t_j^k < \tau _s^1} \right|\Delta t_s^1 = \tau _s^1} \right)d\tau _s^1} \nonumber \\
& \mathop  = \limits^{(a)} \int\limits_0^\infty  {\lambda {e^{ - \lambda \tau _s^1}}\left( {1 - \frac{{\Gamma \left( {k,\lambda \tau _s^1} \right)}}{{\Gamma \left( k \right)}}} \right)d\tau _s^1} \nonumber \\
& = \frac{1}{2} - \sum\limits_{m = 1}^{k - 1} {\frac{{\Gamma \left( {m + 1} \right)}}{{m!{2^{m + 1}}}}} n
 = \frac{1}{{{2^k}}},
\end{align}
where $\Gamma \left( k \right)$ is the Gamma function and $\Gamma \left( {k,x} \right)$ is the upper incomplete gamma function; $(a)$ is obtained as the cdf of Erlang distribution is given by
\begin{align*}
p\left( {\Delta t_j^k < \tau } \right) = F\left( {\tau ,k,\lambda } \right) = 1 - \frac{{\Gamma \left( {k,\lambda \tau } \right)}}{{\Gamma \left( k \right)}};
\end{align*}


Substituting \eqref{tj1} and \eqref{tjk} into \eqref{d_redundant}, $d$ is given by
\begin{align}\label{duplication-value}
d = 2{\sum\limits_{k = 2}^\infty  {{2^{ - k}}} } = 1.
\end{align}
 \end{proof}
\emph{Remarks:} In contrast to the naive protocol, the amount of duplicate packets at each node does not grow with the network size, only being a small constant that can be accurately evaluated.

\subsubsection{Verification through Simulation}
Fig.~\ref{k-prob} compares the simulated and theoretically calculated value of \eqref{k-probability}. $\lambda$ is chosen as $1$. It is shown that in all cases, the theoretic and experimental results match well, which lays the foundation for further estimating the redundant amount.

\begin{figure}[h]
\begin{minipage}[t]{0.49\linewidth}
\centering
\includegraphics[width=1.0\textwidth]{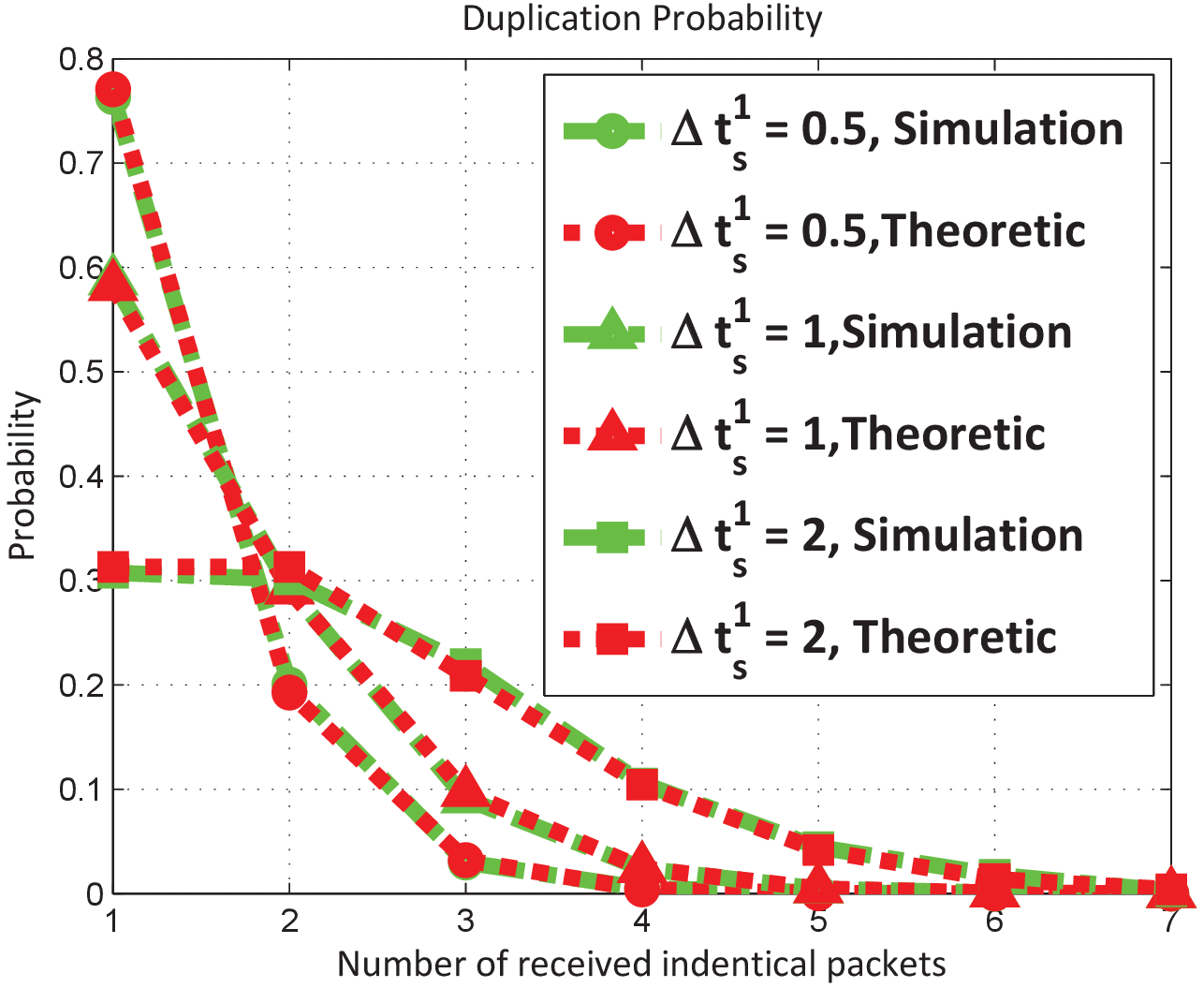}
\caption {Probability of duplication}
\label{k-prob}
\end{minipage}
\hfill
\begin{minipage}[t]{0.49\linewidth}
\centering
\includegraphics[width=1.0\textwidth]{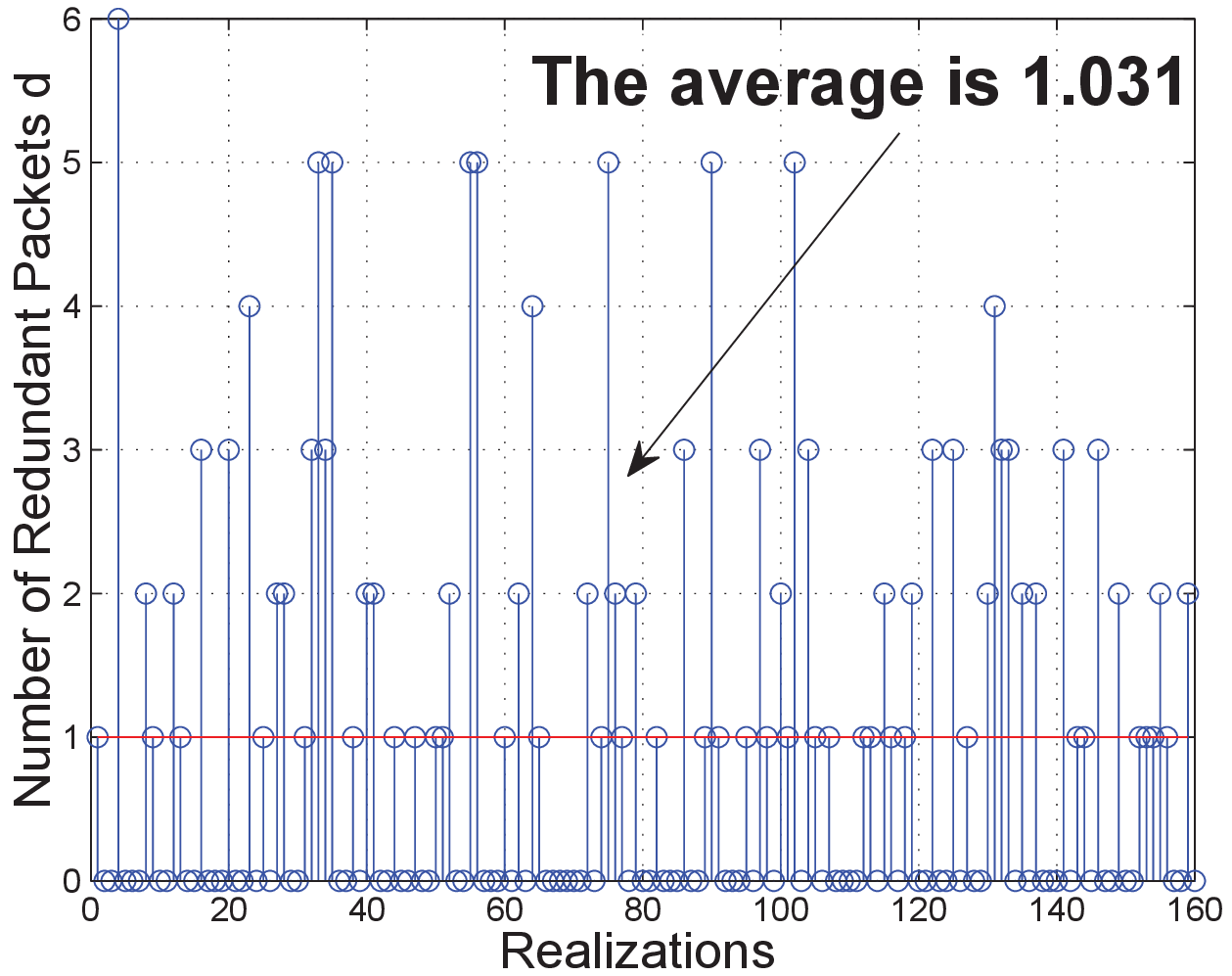}
\caption {Counted redundant packets}
\label{n_redundant}
\end{minipage}
\end{figure}

In Fig.~\ref{n_redundant}, we randomly pick one relay node and one destination node, and simulate the meeting process between them and the source node. For each packet received, we count the number of redundant packets. Among the $150+$ realizations, the number of redundant packets varies from $0$ to $6$. However, the average value is calculated as $1.031$, which confirms the result of Theorem \ref{theorem-duplication}.

\subsection{Analysis of the $l$-packet spreading time}
\noindent To evaluate the rateless-coding-assisted $l$-packet spreading time, we model the spreading as two concurrent processes, as illustrated in Fig.~\ref{StateTransitionDiagramMultiple}. The vertical Markov chain represents the \emph{Relay Initialization} process and the horizontal Markov chains represent the \emph{Packet Collection} processes.
\begin{figure}[h]\center
\includegraphics[width=0.40\textwidth]{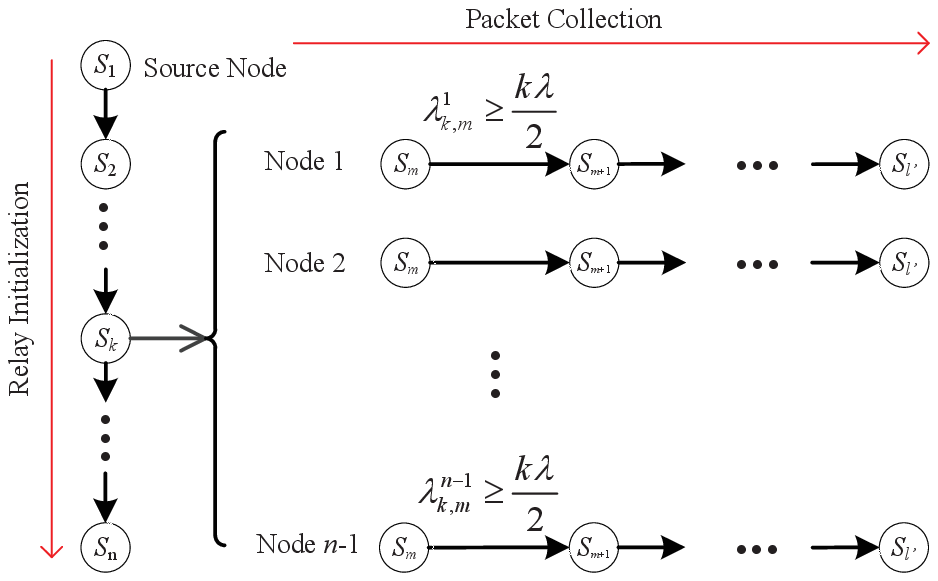}
\caption {State Transition Diagram for Multiple Packet Spreading}
\label{StateTransitionDiagramMultiple}
\end{figure}

At the very beginning, there are $n-1$ destination nodes. \emph{Relay Initialization} means the destination nodes gradually assume dual roles as relay nodes by collecting packets directly from the source. The initialization process starts with no relay nodes and ends with $n-1$ relay nodes (except the source node). State $\tilde{S}_k$ on the vertical chain denotes there are $k$ nodes ($k-1$ relay nodes and the source node) disseminating packets.

During relay initialization, each destination node is also collecting new packets both from the source node and the relay nodes, namely \emph{Packet Collection}. As shown in Fig.~\ref{StateTransitionDiagramMultiple}, the packet collection processes can be viewed as $n-1$ horizontal Markov chains, each corresponds to a destination node. State $S_m$ on the chain denotes  $m$ packets have been collected. After collecting a new packet, the corresponding horizontal chain moves to the next state. The packet collection process stops when every node has received no less than $l'$ \emph{non-duplicate} packets, at which time all receivers can recover the $l$ original source packets with high probability.

When there are $k-1$ relay nodes (i.e. at state $\tilde{S}_k$ in relay initialization), the transitions in the packet collection process can be analyzed as follows. According to Theorem \ref{theorem-duplication}, a destination node takes $\frac 1 2$ probability to move to the next state upon every meeting with a relay node, otherwise directly move to the next state upon a meeting with the source node.

Thus, the overall packet collection rate is
\begin{align*}
\lambda _{k,m}^i = \left\{ {\begin{array}{*{20}{c}}
{\lambda  + \left( {k - 2} \right)\frac{\lambda }{2} = \frac{{k\lambda }}{2},\forall i \in \left\{ {relay\;nodes} \right\},}\\
{\lambda  + \left( {k - 1} \right)\frac{\lambda }{2} = \frac{{\left( {k + 1} \right)\lambda }}{2},\forall i \notin \left\{ {relay\;nodes} \right\},}
\end{array}} \right.
\end{align*}
for each state $S_m$ in the packet collection process.

\emph{Remarks:} the packet collection rate for each nodes is solely controlled by the number of relay nodes in the network: the more relay nodes, the faster a new packet is collected. To simplify the analysis, we assume $\lambda _{k,m}^i \approx \frac{{k\lambda }}{2}$ for all $i$, which will result in a slightly longer spreading time.

\subsubsection{The average number of distinct packets collected in state $\tilde{S}_k$}
The sojourn time $T_k$ for $\tilde{S}_k$ in relay initialization is exponentially distributed, with PDF
\begin{align*}
f\left( \tau  \right) = \left( {n - k} \right)\lambda {e^{ - \left( {n - k} \right)\lambda \tau }}, \quad (\tau >0).
\end{align*}

The number of packets collected in $\tau$ time, denote by $\Delta {l_k^\tau}$, is Poisson distributed, and
\begin{align*}
p\left( {\Delta l_k^\tau  = i} \right) = \frac{{{e^{ - \frac{{k\lambda }}{2}\tau }}{{\left(\frac{{k\lambda }}{2}\tau \right)}^i}}}{{i!}}, \quad (i \in N).
\end{align*}

Denote by $\Delta {l_k}$ the number of packets collected in $\tilde{S}_k$, then the PDF of $\Delta {l_k}$ is derived as
\begin{align}\label{geometric-accum}
p\left( {\Delta {l_k} = i} \right) &= \int\limits_0^\infty  {f\left( \tau  \right)P\left( {\Delta l_k^\tau  = i} \right)d\tau }
= \frac{{2n - 2k}}{{2n - k}}{\left( {\frac{k}{{2n - k}}} \right)^i}.
\end{align}

If we define ${p_k} = \frac{{2n - 2k}}{{2n - k}}$, then the PDF of \eqref{geometric-accum} can be rewritten as $p\left( {\Delta {l_k} = i} \right) = {\left( {1 - {p_k}} \right)^i}{p_k}$, which is a geometric distribution.

The average number of new packets collected by each node in $\tilde{S}_k$ is given by
\begin{align*}
E\left[ {\Delta {l_k}} \right] = \sum\limits_{i = 0}^\infty  {\left( {\frac{{2n - 2k}}{{2n - k}}{{\left( {\frac{k}{{2n - k}}} \right)}^i} \times i} \right)} = \frac{k}{{2\left( {n - k} \right)}}.
\end{align*}

\subsubsection{The average number of distinct packets collected in relay initiation}
\begin{lemma}
For large enough $l$, the average number of distinct packets collected by each node at the end of relay initiation, as denoted by $l_0$, is approximately $\frac{{n\ln n}}{2}$.
\end{lemma}
\begin{proof}
By summing the number of distinct packets collected in each state $\tilde{S}_k$, the total number is given by
\begin{align*}
{l_0} &= \sum\limits_{k = 1}^{n - 1} {\frac{k}{{2\left( {n - k} \right)}}} = \frac{n}{2}\left( {\sum\limits_{k = 1}^{n - 1} {\frac{1}{{n - k}}}  - 1} \right)\nonumber \\
&= \frac{n}{2}\left( {\ln n + \gamma  - 1} \right) + \frac 1 2 + o\left( 1 \right) \approx \frac{{n\ln n}}{2},
\end{align*}
where $\gamma \approx 0.57721$ is the Euler-Mascheroni constant.
\end{proof}

\emph{Discussions:} For not large enough $l$, each node may already obtain $l$ distinct packets before reaching $\tilde{S}_k$ in relay initialization. In the extreme case when $l=1$, it is straightforward that, with the help of relay nodes, the source node doesn't need to meet all $n-1$ nodes to complete the spreading. However, for large enough $l$, the number of packets obtained by each node may not be enough for decoding, thus the spreading continues.

\subsubsection{The average source-to-destination delay}
Let $\tilde{S}_{k^*}$ be the relay initialization state in which each node has received enough packets for decoding, where $k^*$ is the ending state number. Denote by $l\left( {{k^*}} \right)$ the total number of packets received from $\tilde{S}_1$ to $\tilde{S}_{k^*}$. We have
\begin{align}\label{E-l-k*}
&E\left[ {l\left( {{k^*}} \right)} \right] = \sum\limits_{k = 1}^{{k^*}} {E\left[ {\Delta {l_k}} \right]} = \sum\limits_{k = 1}^{{k^*}} {\frac{k}{{2\left( {n - k} \right)}}} \nonumber \\
&= \frac{{n\left( {\ln n - \ln \left( {n - {k^*}} \right)} \right) - {k^*}}}{2} + o\left( 1 \right) \nonumber \\
& = \Theta \left( {n\ln \left( {\frac{n}{{n - {k^*}}}} \right)} \right).
\end{align}
\begin{lemma}\label{Ending-State}
When $l$ is small enough, the ending state number $k^*$ can be numerically obtained by solving $E\left[ {l\left( {{k^*}} \right)} \right] = l$ in \eqref{E-l-k*}. Otherwise, when $l$ is large enough, $k^*=n$.
\end{lemma}
\begin{proof}
The proof is omitted in the interest of space.
\end{proof}

\emph{Definition:} the relay-assisted packet collection delay under the RMPR scheme, as denoted by ${{D_{r,r}}\left( l \right)}$, is the source-to-destination delay for an arbitrarily chosen node to receive enough packets for decoding the $l$ source packets.
\begin{theorem}\label{S2DDelay}
When $l$ is small enough, the average $l$-packet collection delay under the RMPR protocol, as denoted by $E\left[{{D_{r,r}}\left( l \right)}\right]$, is approximately $\frac{{2l + {k^*}}}{{n\lambda }}$; for the special case when $l = o\left( {n\ln n} \right)$, $E\left[{{D_{r,r}}\left( l \right)}\right]$ is estimated in closed-form as $\frac{{2l + 2\sqrt {nl} }}{{n\lambda }}$; when $l$ is large enough, the average delay $E\left[{{D_{r,r}}\left( l \right)}\right]$ is approximately $\frac{{2l}}{{n\lambda }}$.
\end{theorem}
\begin{proof}
When $l$ is small enough, ${{D_{r,r}}\left( l \right)}$ is estimated as the sum of sojourn time from state $\tilde{S}_1$ to $\tilde{S}_{k^*}$, i.e.,
\begin{align}\label{Sum-Sojourn}
E\left[ {{D_{r,r}}\left( l \right)} \right] = \sum\limits_{k = 1}^{{k^*}} {E\left[ {{T_k}} \right]}  = \frac{1}{\lambda }\sum\limits_{k = 1}^{{k^*}} {\frac{1}{{\left( {n - k} \right)}}}.
\end{align}

Combining \eqref{Sum-Sojourn} with \eqref{E-l-k*}, we have
\begin{align*}
E\left[ {l\left( {{k^*}} \right)} \right] = \frac{n}{2}\sum\limits_{k = 1}^{{k^*}} {\frac{1}{{\left( {n - k} \right)}}}  - \frac{{{k^*}}}{2} = \frac{{n\lambda }}{2}E\left[ {{D_{r,r}}\left( l \right)} \right] - \frac{{{k^*}}}{2} = l.
\end{align*}

Thus,
\begin{align}\label{D-case1}
E\left[ {{D_{r,r}}\left( l \right)} \right] = \frac{{2l + {k^*}}}{{n\lambda }},
\end{align}
where $k^*$ is obtained according to \emph{Lemma \ref{Ending-State}}.

Though $k^*$ may not be given in closed-form, in the special case when $l$ is so small that ${k^*} = o\left( n \right)$, we may obtain a closed-form approximate $k^*$ by solving the following equation
\begin{align*}
E\left[ {l\left( {{k^*}} \right)} \right] = \sum\limits_{k = 1}^{{k^*}} {\frac{k}{{2\left( {n - k} \right)}}} \approx \frac{{\sum\limits_{k = 1}^{{k^*}} k }}{{2n}} = \frac{{{{\left( {{k^*}} \right)}^2} + {k^*}}}{{4n}} = l,
\end{align*}
and the solution is ${k^*} = 2\sqrt {nl}  - \frac{1}{2} \approx 2\sqrt {nl}$. Thus the average delay is obtained by substituting ${k^*}$ into \eqref{D-case1}
\begin{align}
E\left[ {{D_{r,r}}\left( l \right)} \right] \approx \frac{{2l + 2\sqrt {nl} }}{{n\lambda }}.
\end{align}

When $l$ is large enough, the spreading is not finished after the relay initialization process is finished. ${{D_{r,r}}\left( l \right)}$ is thus constituted of two parts: the relay initialization complete time, as denoted by ${T_{complete}}$, and the extra packet collection time for the remaining $l-l_0$ packets, as denoted by ${T_{residual}}$. Thus,
\begin{align*}
{D_{r,r}}\left( l \right) = {T_{complete}} + {T_{residual}}.
\end{align*}

It is easy to argue that the relay initialization time ${T_{complete}}$ is the same as the single packet spreading time \emph{without} relaying, which is derived in \cite[Section III.A]{tech-report}, only that here the packets received by the nodes may be different.

According the protocol, the new-packet inter-arrival time ${T_m}$ for each node after reaching $\tilde{S}_n$ is independent and identically distributed exponential variable, i.e., ${T_m} \sim Exp\left( \frac{{n\lambda }}{2} \right),\forall m \in l_0, \cdots ,l' - 1$. Therefore $T_{residual}$ is Erlang distributed
\begin{align*}
T_{residual} \sim {\rm{Erlang}}\left( {l' - l_0 - 1,\frac{{n\lambda }}{2} } \right).
\end{align*}

The mean and variance of $T_{residual}$ are given by
\begin{align*}
E\left[ {{T_{residual}}} \right] = \frac{{2\left( {l' - {l_0} - 1} \right)}}{{n\lambda }}, ~~ \sigma_{T_{residual}}^2= \frac{4\left({l' -l_0 - 1}\right)}{{{n^2\lambda ^2}}},
\end{align*}
respectively. Thus when $l$ is large enough, the overall average delay is
\begin{align}
E\left[ {{D_{r,r}}\left( l \right)} \right] &= E\left[ {T_{complete}} + {T_{residual}} \right]  \nonumber \\
&= \frac{{\ln n}}{\lambda } + \frac{{2\left( {l' - {l_0} - 1} \right)}}{{n\lambda }} \approx \frac{{2l}}{{n\lambda }}.
\end{align}
\end{proof}

\subsubsection{The average network spreading time} we now move on to characterize the average source-to-network spreading time.

\emph{Definition:} the relay-assisted $l$-packet spreading time under the RMPR scheme, as denoted by ${{T_{r,r}}\left( l \right)}$, is the time that every node has reached the final state $S_{l'}$ in packet collection.
\begin{theorem}\label{S2NSpreading}
The average $l$-packet spreading time under the RMPR protocol, as denoted by $E\left[{{T_{r,r}}\left( l \right)}\right]$, is upper bounded by $\frac {\ln n}{\lambda}$ and lower bounded by $E\left[ {{D_{r,r}}\left( l \right)} \right]$ for small enough $l$. For large enough $n$ and $l>l_0$, where $l_0=\frac {n\ln n}{2}$, $E\left[{{T_{r,r}}\left( l \right)}\right]$ is approximately $\frac{{2l + 2\sqrt {2\left(l-l_0\right)\ln n} }}{{n\lambda }}$.
\end{theorem}
\begin{proof}
When $l$ is small enough, each node can collect enough packets for recovering before the end of relay initialization, therefore the $l$-packet spreading time is upper bounded by $\frac {\ln n}{\lambda}$. Of course, it should also be larger than the average source-to-destination delay.

When $l$ is large enough, ${{T_{r,r}}\left( l \right)}$ is comprised of the relay initialization time and the maximum of $n-1$ i.i.d. Erlang distributed extra packet collection time $T_{residual}$.
\begin{align*}
{T_{r,r}}\left( l \right) = T_{complete} + \mathop {\max }\limits_{i \in 1, \cdots ,n-1} T_{residual}^i,
\end{align*}
where $T_{residual}^i$ is the Erlang distributed extra packet collection time for node $i$.

The Erlang distribution is a special case of the Gamma distribution, i.e., ${\rm{Erlang}}\left( {k,\lambda } \right) \Leftrightarrow \Gamma \left( {k,\frac{1}{\lambda }} \right)$. For large $k$, the Gamma distribution $\Gamma \left( {k,\frac{1}{\lambda }} \right)$ converges to a Gaussian distribution with mean $\mu  = \frac k \lambda$ and variance ${\sigma ^2} = \frac k {\lambda ^2}$. That is to say, for large $l'$
\begin{align}
T_{residual}^i &\sim {{\cal N}}\left( {\frac{2\left({l' - l_0 - 1}\right)}{n\lambda }, \frac{4\left({l' - l_0 - 1}\right)}{{{n^2\lambda ^2}}}} \right), \nonumber \\
&\forall i=\left\{1, \cdots , n-1\right\}.
\end{align}

By this approximation, the problem simplifies to estimating the maximum of $n-1$ i.i.d. Gaussian variables. Thus,
\begin{align}
&{T_{r,r}}\left( l \right) = T_{complete}  \nonumber \\
& + \mathop {\max }\limits_{i \in 1, \cdots ,n - 1} \left\{ {{\frac{{2\left( {l'  - l_0 - 1} \right)}}{{n\lambda }}} + \sqrt {{\frac{{4\left( {l' - l_0 - 1} \right)}}{{{n^2}{\lambda ^2}}}}} {X_i}} \right\},
\end{align}
where $X_i \sim {\cal N}\left(0,1\right), i \in 1, \cdots ,n-1$ are i.i.d unit Gaussian variables.

\begin{lemma}[\cite{extreme}]\label{MaxGaussian}
If ${X_i} \sim {\cal N}\left(0,1\right),i \in {1, \cdots ,n}$ are a series of i.i.d. Gaussian variables, and ${M_n} = \mathop {\max }\limits_{i = 1, \cdots ,n} \left\{X_i\right\}$ is the maximum of the $n$ Gaussian variables. Then, we have ${M_n} \sim \sqrt {2\ln n}$ with high probability when $n$ is large.
\end{lemma}

According to \emph{Lemma \ref{MaxGaussian}},
\begin{align}
&E\left[{T_{r,r}}\left( l \right)\right] = E\left[T_{complete}\right] \nonumber \\
&+ \frac{{2\left( {l'  - l_0 - 1} \right)}}{{n\lambda }} + \frac{{2\sqrt {l' - l_0 - 1} }}{{n\lambda }}\mathop {\max }\limits_{i = 1, \cdots ,n - 1} \left\{ {{X_i}} \right\} \nonumber \\
& \approx \frac{{\ln n}}{\lambda } + \frac{{2\left( {l' - l_0- 1} \right) + 2\sqrt {2\left( {l' - l_0- 1} \right)\ln \left( {n - 1} \right)} }}{{n\lambda }}.
\end{align}

Since $l'=(1+\epsilon)l \approx l$ and ${l_0} \approx \frac{{n\ln n}}{2}$, for large $l$, $n$ and $l > l_0$, $E\left[{T_{r,r}}\left( l \right)\right]$ is approximated by
\begin{align}
E\left[{T_{r,r}}\left( l \right)\right] \approx \frac{{2l + 2\sqrt {2\left(l-l_0\right)\ln n} }}{{n\lambda }}.
\end{align}
\end{proof}


\subsubsection{Discussions}
We now compare the results in \emph{Theorem \ref{S2DDelay}} and \emph{Theorem \ref{S2NSpreading}} with their non-rateless counterparts in \cite[Section IV.B]{tech-report}. For large-size message, the rateless-coding-assisted scheme significantly reduce both source-to-destination delay and source-to-network spreading time by at least a factor of $\ln(n)$. In essence, the proposed method exploits the strength of rateless codes from point-to-point transmissions and extends its application to point-to-network scenarios.


\section{Simulation Results}
\noindent The $n$ mobile nodes are deployed on a unit square and follow the HSMNet model described in Section \ref{ModelSection}. Without loss of generality, we only simulated the Random Direction mobility model. As for other HSMNets, the only difference relies in $\lambda$.

\begin{figure}[h] \centering
\includegraphics[width=0.45\textwidth]{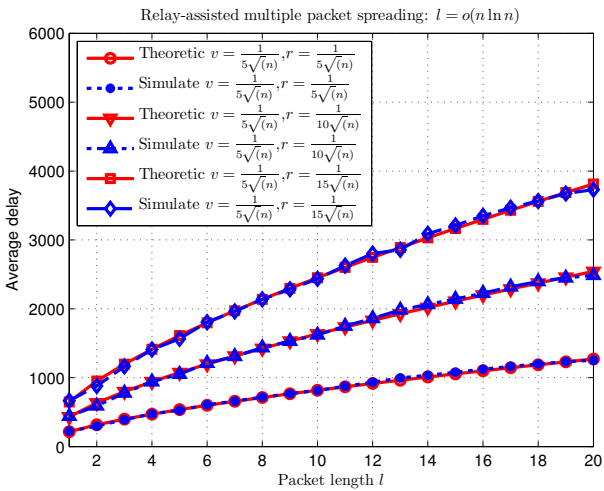}
\caption {Multiple packet delay with relaying when $l=o(n \ln n)$}
\label{MultiWithRelay_ShortPacket}
\end{figure}

\begin{figure}[h] \center
\includegraphics[width=0.5\textwidth]{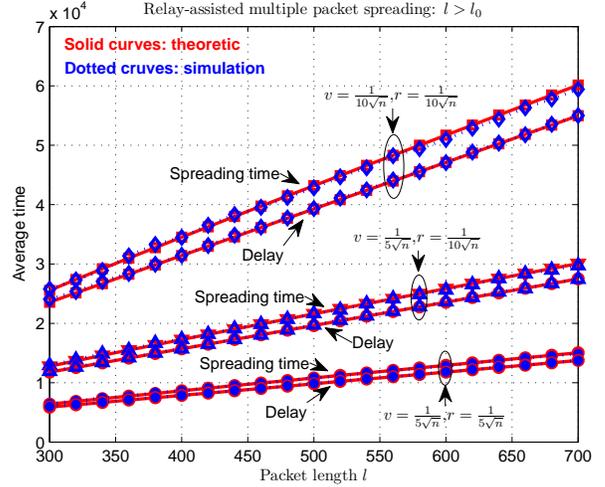}
\caption {Multiple packet delay with relaying when $l$ is large enough}
\label{MultiWithRelay_LongPacket}
\end{figure}

The multiple packet spreading with relaying case is shown in Fig.~\ref{MultiWithRelay_ShortPacket}-\ref{MultiWithRelay_LongPacket}. The theoretic results and the simulation results are plotted in red curves and blue curves, respectively. The former shows the multi-packet delay when $l=o(n \ln n)$, and the latter shows both the multi-packet delay and spreading time when $l$ is large. It is shown that all simulation results match the theoretic analysis perfectly. When $l$ grows large, the delay gradually becomes linear with $l$.

%
%
%

\section{Conclusions}
\noindent In this paper, we study \emph{multiple packet broadcasting employing rateless codes}. Our results include both \emph{point-to-point} delay and the \emph{point-to-network} spreading time. It is shown that, the \emph{Rateless-coding-assisted Multi-Packet Relaying (RMPR) scheme} can significantly reduce packet duplication, which not only makes multi-packet relaying possible but also greatly simplify the implementation. Finally, extensive simulations are conducted to support our theoretical analysis.

\bibliographystyle{unsrt}


\newpage

\end{document}